\documentclass[leqno]{amsart}

\usepackage[round,comma,numbers]{natbib}                

\usepackage{amsmath, amsthm}
\usepackage{amsfonts,amssymb,dsfont}
\usepackage{nicefrac,mathrsfs}
\usepackage{bm}                                 
\usepackage[backgroundcolor=white,bordercolor=orange]{todonotes}
\usepackage{bbm}

\newcommand{\ccA}{{\mathscr A}}

\newcommand{\ccF}{{\mathscr F}}
\newcommand{\ccG}{{\mathscr G}}

\newcommand{\ccP}{{\mathscr P}}

\newcommand{\Ind}{{\mathds 1}}
\newcommand{\ind}[1]{\Ind_{\{#1\}}}

\newcommand{\bbF}{\mathbb{F}}

\newcommand{\bbG}{\mathbb{G}}

\newenvironment{enumeratei}
  {\begin{enumerate} }
  {\end{enumerate}}

\newtheorem{theorem}{Theorem}[section]
\newtheorem{lemma}[theorem]{Lemma}              
\newtheorem{proposition}[theorem]{Proposition}  

\theoremstyle{definition}

\newtheorem{definition}{Definition}[section]

\setcitestyle{square}
\usepackage{eurosym}
\DeclareRobustCommand{\officialeuro}{%
	\ifmmode\expandafter\text\fi
	{\fontencoding{U}\fontfamily{eurosym}\selectfont e}}

\usepackage[pdftex,colorlinks,urlcolor=red,citecolor=blue,linkcolor=red]{hyperref}
\setcitestyle{square}
\usepackage{comment}

\begin{document}

\title[Credit risk under ambiguity]{Ambiguity in defaultable term structure models.}

\author{Tolulope Fadina}
\address{Freiburg University, www.stochastik.uni-freiburg.de, Email: tolulope.fadina@stochastik.uni-freiburg.de}
\author {Thorsten Schmidt}
    \address{Freiburg University, www.stochastik.uni-freiburg.de and Freiburg Research Institute of Advanced Studies (FRIAS) and University of Strasbourg Institute for Advanced Study (USIAS). Email:  thorsten.schmidt@stochastik.uni-freiburg.de}
\thanks{Financial support by Carl-Zeiss-Stiftung is gratefully acknowledged. We thank Monique Jeanblanc for her generous support and helpful comments. }

\maketitle

\begin{abstract}
\noindent 
We introduce the concept of no-arbitrage in a credit risk market under
ambiguity considering an intensity-based framework. We assume the
default intensity is not exactly known but lies between an upper and
lower bound. By means of the Girsanov theorem, we start from the
reference measure where the intensity is equal to $1$ and construct the
set of equivalent martingale measures. From this viewpoint, the credit
risky case turns out to be similar to the case of drift uncertainty in
the $G$-expectation framework. Finally, we derive the interval of
no-arbitrage prices for general bond prices in a Markovian setting.


\end{abstract}

\newcommand{\ulambda}{\underline{\lambda}}
\newcommand{\olambda}{\overline{\lambda}}
\newtheorem{thm}{Theorem}
\newtheorem{remark}[thm]{Remark}

\newtheorem{assumption}{Assumption}
\newtheorem{case}{Case}
\newtheorem{nota}[thm]{Notational convention}
\newcommand {\PP}{\mathbb P}

\newcommand {\QQ}{\mathbb Q}
\newcommand {\half}{\frac{1}{2}}

\newcommand {\fP}{\mathfrak{P}}
\newcommand\norm[1]{\left\lVert#1\right\rVert}

\section{Introduction}
A critical reflection on the current financial models reveals that models for financial markets require the precise knowledge of the underlying probability distribution, which is clearly unknown. 
Typically, the unknown distribution is either estimated by statistical methods or calibrated to given market data by means of a model for the financial market. 
The analysis of the recent financial crisis suggests that this introduces a large {\em model risk}. Already, \cite{Knight1921} pointed towards a formulation of risk which is able to treat such challenges in a systematic way. He was followed by \cite{Ellsberg1961}, who called random variables with known probability distribution \emph{certain}, and those where the probability distribution is not known as \emph{uncertain}. Following the modern literature in the area, we will call the feature that the probability distribution is not entirely fixed, \emph{ambiguity}.  This area has recently renewed the attention of researchers in mathematical finance to fundamental subjects such as arbitrage conditions, pricing mechanisms, and super-hedging. 
 Roughly speaking, ambiguity focuses on a set of probability measures whose role is to determine events that are relevant and those that are negligible.   
In this paper, we introduce the concept of ambiguity to \emph{term structure models}. The starting point for term structure models are typically  bond prices of the form
\begin{align}
\label{eq:bond_price}
    P(t,T) = e^{-\int_t^T f(t,u) du}, \qquad 0 \le t \le T
   \end{align}
where $(f(t,T))_{0\leq t\leq T}$ is the instantaneous forward rate and $T$ is the maturity time. This follows the seminal approach proposed in \cite{HJM}. The presence of credit risk\footnote{The risk that an agent fails to fulfil contractual obligations. Example of an instrument bearing credit risk is a corporate bond.} in the model introduces an additional factor known as the default time. In this setting, bond prices are assumed to be absolutely continuous with respect to the maturity of the bond. This assumption is typically justified by the argument that, in practice, only a finite number of bonds are liquidly traded and the full term structure is obtained by interpolation, thus is smooth. There are two classical approaches to model market default risk:  \emph{Structural approach} \cite{Merton1974} and the \emph{Reduced-form approach} (see for example, \cite{ArtznerDelbaen95, DuffieSchroderSkiadas96, Lando94} for some of the first works in this direction). 

In structural models of credit risk, the underlying state is the value of a firm's assets which is assume to be observable. Default happens at maturity time of the issued bond if the firm value is not sufficient to cover the liabilities. Hence, default is not a surprise. One exception is the structural model of \cite{Zhou}, in which the value of the firm’s assets is allowed to jump. In fact, the value of the firm’s assets is not observable. A credit event usually occurs in correspondence of a missed payment by a corporate entity and, in many cases, the payment dates or coupon dates are publicly known in advance. For example, the missed coupon payments by Argentina on a notional of ${\$}29$ billion (on July $30, 2014$), and by Greece on a notional of $\euro 1.5$ billion (on June $30, 2015$).

Reduced-form (HJM-type) models for defaultable term structure generally assume the existence of a default intensity which implies that default occurs with probability zero at a predictable time. Consequently, reduced-form models typically postulate that default time is totally inaccessible and prior to default, bond prices are absolutely continuous with respect to the maturity.  That is, under the assumption of zero \emph{recovery}\footnote{The amount that the owner of a defaulted claim receives upon default.}, credit risky bond prices $P(t,T)$ is given by 
\begin{align} \label{PtT1}  P(t,T) = \mathbb{I}_{\{ \tau > t\} } e^{-\int_t^T f(t,u) du}      
\end{align}
with $\tau$ denoting the \emph{random default time}. This approach has been studied in numerous works and up to a great level of generality, see \cite[Chapter 3]{eberlein-oezkan-03}, for an overview of relevant literature. 
The random default time $\tau$ is assumed to have an intensity process $\lambda$. For example, with constant intensity $\lambda$, default has a Poisson arrival at intensity $\lambda$ . More generally, for $\tau > t$, $\lambda_t$ may be viewed as the conditional rate of arrival of default at time $t$, given information up to that time. 
In a situation where the owner of a defaultable claim recovers part of its initial investment upon default, the associated survival process $\mathbb{I}_{\{ \tau > t\} }$ in \eqref{PtT1}, is replaced by a semimartingale.

Under ambiguity, we suggests there is some a prior information at hand which gives a upper and lower bounds on the \emph{intensity}. It seems that the market has  acknowledged uncertainty in this factors for a long time because there are important sources of additional information  available. The implicit assumption that the probability distribution of default is known is quite sensitive. Thus, we analyse our problem in a multiple priors model which describe uncertainty about the correct probability distribution. By means of the Girsanov theorem, we constructs the set of priors from the reference measure. The assumption is that all priors are equivalent, or at least absolutely continuous with respect to the reference measure.

In view of our framework, it is only important to acknowledge that a rating class provides an estimate of the one-year default probability in terms of a confidence interval. Also estimates for 3-, and 5-year default probabilities can be obtained from the rating migration matrix. Thus, leading to a certain amount of model risk.

The aim of this paper is to incorporate this model risk into our models. That is, to provide a framework for modeling defaultable term structure models taking into account model risk.

The main results are as follows: we obtain a necessary and sufficient condition for a reference probability measure to be a local martingale measure for financial market under ambiguity composed by all credit risky bonds with prices given by \eqref{PtT1}, thereby ensuring the absence of arbitrage in a sense to be precisely specified below. Furthermore, we consider the case where we have partial information on the amount that the owner of a defaulted claim receives upon default. Under the assumption of no-arbitrage, we derive the interval of bond prices in a Markovian setting.
 
This paper is set up as follows: the next section introduces homogeneous ambiguity, and its example. Section \ref{sec:NAA} introduces the fundamental theorem of asset pricing (FTAP) under homogeneous ambiguity. In section \ref{sec:Defaultable}, we derive the robust no-arbitrage conditions for defaultable term structure models with zero-recovery, and fractional recovery of market value. Section \ref{sec:Robust} discusses the bond pricing intervals under the assumption of no-arbitrage opportunities in a Markovian setting.

\section{Ambiguity}
\label{sec:Ambiguity}
 We consider throughout a fixed finite time horizon $T^*>0$. Let $(\Omega, \ccF)$ be a measurable space. 
By \emph{ambiguity} we refer to a set of  probability measures $\ccP$ on the measurable space $(\Omega, \ccF)$. In particular there is no fixed and known measure $P$. 
For credit risk, the most important case is the following case of homogeneous ambiguity: the ambiguity is called \emph{homogeneous} if there is a measure $P'$ such that  $P \sim P'$ for all $P \in \ccP$.
The reference measure $P'$ has the role of fixing events of measure zero for all probability measures under consideration. 
Intuitively, there is no ambiguity on these events of measure zero. We write $E'$  for  the expectation with respect to the reference measure $P'$. 

\begin{remark}
As a consequence of the equivalence of all probability measures ${P \in \ccP}$, all equalities and inequalities will hold almost-surely with respect to any probability measure $P \in \ccP$, or, respectively, to $P'$. 
\end{remark}

\subsection*{Ambiguity in intensity-based models}
Intensity-based models are one of the most frequently used approaches in credit risk, see \cite[Chapter 8]{BieleckiRutkowski2002} for an overview of relevant literature, and now we introduce ambiguity in this class. Consider a probability space $(\Omega,\ccG,P')$ supporting a $d$-dimensional Brownian motion $W$ with canonical and augmented\footnote{Augmentation can be done in a standard fashion with respect to $P'$.} filtration $\mathbb{F}= (\ccF_t)_{0\leq t \leq T^*}$ and a standard exponential random variable $\tau$, independent of $\ccF_{T^*}$, that is, $P'( t< \tau|\ccF_t)= \exp(- t)$, $0\leq t \leq T^*$. The full-filtration $\mathbb G=(\ccG_t)_{0 \le t \le T^*}$ is obtained by a progressive enlargement\footnote{We refer to \cite{AksamitJeanblanc} for further literature.} of $\mathbb F$ with $\tau$, i.e.
\begin{align*}
	\ccG_t = \bigcap_{\epsilon >0}\sigma( \mathbbm{1}_{\{t \ge \tau\}} , W_s: 0 \le s \le t+\epsilon), \qquad 0 \le t \le T^*.
\end{align*}
We assume that $\ccG=\ccG_{T^*}$. By means of the Girsanov theorem, we explicitly construct the measures $P^\lambda$ where under $P^\lambda$, the default time $\tau$ has intensity $\lambda$. In this regard, consider progressively measurable and positive processes $\lambda$ and define density processes $Z^{\lambda}$ by  
\begin{align}\label{def:density}
	Z^{\lambda}_t := \begin{cases}
	\exp\Big( \int_0^t (1-\lambda_s) ds \Big),  & t < \tau \\
	\lambda_\tau \exp\Big( \int_0^\tau (1-\lambda_s) ds \Big) & t \ge \tau.
	\end{cases}
\end{align}   
Note that $Z^\lambda$ is indeed a $\bbG$-martingale and corresponds to a Girsanov-type change of measure (see Theorem VI.2.2 in \cite{Bremaud1981}).
If moreover $E'[Z^{\lambda}_{T^*}]=1$ we may define the measure $P^\lambda \sim P'$ as 
\begin{align} \label{Girsanov}
P^\lambda(A):= {E'}(\mathbbm{1}_{A}Z^{\lambda}_{T^*}) \quad \forall A\in \ccG.\end{align}

The degree of ambiguity in this setting will be measured in terms of an interval $[\underline{\lambda}, \overline{\lambda}] \subset (0,\infty) $ where $ \underline{\lambda}$ and $\overline{\lambda}$ denote lower (upper) bounds in the default intensity.  
We define the set of density generators $\bar H$ by
$$ \bar H := \{ \lambda: \lambda \text{ is }\mathbb{F}\text{-predictable and } \underline{\lambda} \le \lambda_t \le \overline \lambda, \quad t \in [0,T^*] \}.$$

\noindent Additionally, we denote the set of probability measures under ambiguity on the default intensity by
\begin{align}
	\bar\ccP :=\{ P^\lambda: \lambda \in \bar H \}.
\end{align}

\begin{remark}This setting can easily be extended to time varying boundaries\footnote{See \cite{Riedel} for a discussion of the concept of dynamic consistency in dynamic models.} $[\underline{\lambda}(t),\overline{\lambda}(t)]$, $0 \le t \le T^*$. Also the extension to random processes is possible, however at the expense of some delicate measurability issues, confer \cite{Marcel}.
\end{remark}

\begin{lemma}\label{lemma1}
$\bar\ccP$ is a convex set.  
\end{lemma}
\begin{proof}
Consider $P^{\lambda'}, P^{\lambda''} \in \bar\ccP$ and $\alpha \in (0,1)$. Then,
\begin{align*}
\alpha P^{\lambda'}(A) + (1-\alpha)  P^{\lambda''}(A) &= E'\big[ \Ind_A (\alpha Z_{T^*}^{\lambda '} + (1-\alpha) Z_{T^*}^{\lambda''})\big].
\end{align*}
Now consider the (well-defined) intensity $\lambda$, given by
$$ \int_0^t \lambda_s ds := t-\log\Big[ \alpha e^{\int_0^t (1-\lambda'_s)ds } + (1-\alpha) e^{\int_0^t (1-\lambda_s'')ds} \Big], $$
$0 \le t \le T^*$. Then,
$$ 
	\alpha Z_{T^*}^{\lambda '} + (1-\alpha) Z_{T^*}^{\lambda''} = Z_{T^*}^\lambda $$
such that by \eqref{Girsanov}, $P^\lambda \sim P'$ refers to a proper change of measure. We have to check that $\lambda \in \bar H$, which means that $\lambda$ satisfies $\lambda \in [\underline \lambda, \overline \lambda]$, $0 \le t \le T^*$: note that
\begin{align*}
\lefteqn{ t-\log\Big[ \alpha e^{\int_0^t (1-\lambda'_s)ds } + (1-\alpha) e^{\int_0^t (1-\lambda_s'')ds} \Big]} \hspace{2cm} \\
& \le t-\log\Big[ \alpha e^{\int_0^t (1-\overline \lambda )ds } + (1-\alpha) e^{\int_0^t (1-\overline \lambda)ds} \Big] \\
& \le t- t (1-\overline \lambda) = \overline \lambda t,
\end{align*}
and $\lambda_s \le \overline \lambda$ follows. In a similar way we obtain $\lambda_s \ge \underline \lambda$.
\end{proof}

\begin{remark}
 Intuitively, the requirement $\underline{\lambda}>0$ states that there is always a positive risk of experiencing a default, which is economically reasonable. Technically it has the appealing consequence that all considered measures in $\bar\ccP$ are equivalent.
\end{remark}
It turns out that the set of possible densities will play an important role in connection with measure changes. In this regard, we define  \emph{admissible measure changes} with respect to $\bar \ccP$ by
$$ \bar \ccA:= \{ \lambda^*: \lambda^* \text{ is }\bbF\text{-predictable and }E^{P}[Z^{\lambda^*}_{T^*}]<\infty\text{ for all }P \in \bar \ccP\}. $$
The associated Radon-Nikodym derivatives $Z^{\lambda^*}_{T^*}$ for $\lambda^*\in\bar \ccA$ are the possible Radon-Nikodym derivatives for equivalent measure changes when starting from a measure $P\in\bar \ccP$.

\section{Absence of arbitrage under homogeneous ambiguity}
\label{sec:NAA}
Absence of arbitrage and the respective generalization, \emph{no free lunch} (NFL), \emph{no free lunch with vanishing risk} (NFLVR), are well-established concept under the assumption that the probability measure is known and fixed. Here we give a small set of sufficient conditions for absence of arbitrage extended to the setting with homogeneous ambiguity and directly formulated in terms of bond markets. 

For the beginning we consider a bond market consisting only of finitely many traded bonds, small market, an extension to a more general case follows below.

Consider, as previously, a (general) set of probability measures $\ccP$ on the measurable space $(\Omega,\ccG)$ where $P'$ is the dominating measure, i.e.\ $P\sim P'$ for all $P \in \ccP$. Recall that, there is a filtration $\bbG$ satisfying the usual conditions with respect to $P'$.
Discounted price processes are given by a finite dimensional semimartingale $X$ with respect to $\bbG$. The semimartingale property holds equivalently in any of the filtration $\mathbb{G}_+$ or the augmentation of $\mathbb{G}_+$, see \cite[Proposition 2.2]{NeufeldNutz:Measurability}. It is well-known that then $X$ is a semimartigale for all $P\in \ccP$.

Self-financing trading strategies are given by predictable and $X$-integrable processes $\Phi$ and the discounted gains process is 
given by the stochastic integral of $\Phi$ with respect to $X$, as denoted by
$$  (\Phi \cdot X)_t= \int_0^t \Phi_u  dX_u.  $$
An \emph{arbitrage} is a strategy which starts from zero initial wealth, has non-negative pay-off under all possible future scenarios, hence for all $P \in \ccP$, where there is at least one $P$ such that the pay-off is positive. This is formalized in the following definition, compare for example \cite{Vorbrink14}. As usual a trading strategy is $a$-admissible, if $(\Phi \cdot X)_t \geq -a$ for all $0 \le t \le  T^*$.
 \begin{definition}\label{def:arbitrage}
 A self-financing trading strategy $\Phi$ is called \emph{$\ccP$-arbitrage} if it is $a$-admissible for some $a>0$ and
 \begin{itemize}
 \item $(\Phi \cdot X)_{T^*}  > 0$ \hspace{11.9mm} for all $P \in  \ccP$,
 \item $P((\Phi \cdot X)_{T^*} > 0)  > 0$ \ for every $P \in \ccP.$
 \end{itemize}
 \end{definition}
This describes the possibility of getting arbitrarily rich with positive probability by taking small or vanishing risk. A probability measure $Q$ is called \emph{local martingale measure}, if $X$ is a $Q$-local martingale.

It is well-known that \emph{no arbitrage} or, more precisely, \emph{no free lunch with vanishing risk} (NFLVR) in a general semimartingale market is equivalent to the existence of an equivalent local martingale measure (ELMM), see \cite{DelbaenSchachermayer1994,DelbaenSchachermayer1998}. The technically difficult part of this result is to show that a precise criterion of absence of arbitrage implies the existence of an ELMM. In the following we will not aim at such a deep result under ambiguity, but utilize the easy direction, namely that existence of an ELMM implies the absence of arbitrage as formulated below.

The definition of ELMMs with respect to $\ccP$ simplifies because we are considering the homogenous case with dominating measure $P'$.
\begin{definition}
The measure $Q$ is called \emph{equivalent local martingale measure}, if  $Q \sim P'$  and $Q$ is a local martingale measure. 
\end{definition} 

From the classical \emph{fundamental theorem of asset pricing} (FTAP), the following result follows easily.
\begin{theorem}
\label{thm:FTAP-NFLV_homogeneous}
If there exists an equivalent local martingale measure $Q$ for the homogeneous family $\ccP$, then there is no arbitrage in the sense of Definition \ref{def:arbitrage}.
\end{theorem}
\begin{proof}
Indeed, assume there is an arbitrage $\Phi$ with respect to some measure $P\in \ccP$.
By definition,  $Q\sim P$ and so $Q$ is an ELMM for $P$. But then $\Phi$ would be an arbitrage strategy with existing ELMM $Q$, a contradiction to the classical FTAP.
\end{proof}

\section{Defaultable term structures under ambiguity}
\label{sec:Defaultable}
In this section we consider dynamic term structure modelling under default risk when there is ambiguity about the default intensity. 
The relevance of this issue has, for example, already been reported in \cite{Riedel2015}. Here we take this as motivation to propose a precise framework taking ambiguity on the default intensity into account. We continue to work in the setting introduced in Section \ref{sec:Ambiguity}.

\subsection{Dynamic term structures}
We define the \emph{default indicator process}  $H$  by 
$$ H_t =  \mathbbm{1}_{\{ t \geq \tau \}}, \qquad 0 \le t \le T^*. $$
The associated \emph{survival process}  is $1- H.$
A credit risky bond with a maturity time $T$ is a contingent claim promising to pay one unit of currency at $T$. We denote the price of such a bond  at time $t \leq T$ by $P(t,T)$. If no default occurs prior to $T$, $P(T,T)=1$. We will first consider \emph{zero recovery}, i.e., assume that the bond loses its total value at default. Then $P(t,T)=0$ on $\{t \geq \tau\}$.  

Besides zero recovery we only make the weak assumption that bond-prices prior to default are positive and absolutely continuous with respect to maturity $T$. This follows the well-established approach by \cite{HJM}. In this regard, assume that 
\begin{equation}
\label{defaultable price}
P(t, T) = \mathbbm{1}_{\{\tau > t\}} \exp\left(-\int_{t}^{T} f(t,u)du \right) \qquad  0 \leq t \leq T \le T^*.
\end{equation}
The initial forward curve $T \mapsto f(0,T)$ is then assumed to be sufficiently integrable and the \emph{forward rate processes} $f(\cdot,T)$ follow It\^o processes satisfying 
$$f(t,T)= f(0,T)+ \int_{0}^{t}a (s,T)ds + \int_{0}^{t}b (s,T)dW_s,$$
for $0 \leq t \leq T \le T^*$. In principle, $T \ge T^*$ would be possible to consider without additional difficulties.
\begin{assumption}
\label{parameter}
We require the following technical assumptions:
\begin{itemize}
\item[(i)] the initial forward curve is measurable, and integrable on $[0,T^*]$:
$$\int_{0}^{T^*}|f(0,u)|du< \infty,$$
\item[(ii)] the drift parameter $a(\omega, s, t)$ is $\mathbb{R}$-valued $\mathcal{O} \otimes \mathcal{B}$-measurable and integrable on $[0,T^*]$:
$$ \int_{0}^{T^*} \int_{0}^{T^*} |a (s,t)| ds dt < \infty, $$
\item[(iii)] the volatility parameter $b(\omega, s, t)$ is $\mathbb{R}^d$-valued, $\mathcal{O} \otimes \mathcal{B}$-measurable, and
$$\sup_{s,t \leq T^*}\| b (s,t)\| < \infty . $$
\end{itemize}
\end{assumption}

Set for $0 \le t \le T \le T^*$,
\begin{align*}
\overline{a}(t,T) &= \int_{t}^{T} a (t,u)du, \\
\overline{b}(t,T) &= \int_{t}^{T} b (t,u)du.
\end{align*}

\begin{lemma} Under Assumption \ref{parameter} it holds that,  
$$ \int_{t}^{T} f (t,u)du =  \int_{0}^{T} f(0,u)du + \int_{0}^{t} \overline{a}(\cdot, u) du + \int_{0}^{t} \overline{b}(\cdot,u) dW_u - \int_{0}^{t} f (u,u)du $$
for $0 \le t \le T \le T^*$, almost surely.
\end{lemma}
This follows as in \cite{HJM}, see for example Lemma 6.1 in \cite{Filipovic2009}.

\subsection{Absence of arbitrage with ambiguity on the default intensity}

We start by stating the classical ingredient to absence of arbitrage in intensity-based dynamic term structure models. Note that, under $P^{\lambda}$, the compensator or the dual predictable projection $H^p$ of $H$ is given by $H^p_t = \int_0^{t\wedge \tau} \lambda_s ds$. By the Doob-Meyer decomposition,
$$ M^{\lambda} := H - H^p, \qquad 0 \le t \le T^* $$
is $P^{\lambda}$-martingale. 

For discounting, we use the bank account. Its value is given by a stochastic process starting with $1$ and we assume a short-rate exists, i.e., the value process of the bank account is $\gamma(t) = \exp(\int_0^t r_s ds)$ with an $\bbF$-predictable process. We assume that $P'(\int_0^{T^*} r_s ds < \infty)=1$. Then, we obtain the following result.

\begin{proposition} 
\label{pro:no-arbitrage}
Consider a measure $Q$ on $(\Omega,\ccG_{T^*})$ with $Q \sim P'$. Assume that Assumption \ref{parameter} holds and $M^\lambda$ is a $Q$-martingale. Then $Q$ is a local martingale measure if and only if 
\begin{enumeratei}
\item $f(t,t)= r_t + \lambda_t, $
\item the drift condition 
	\begin{align*}
 		\bar{a} (t,T)= \frac{1}{2} \left\| \overline{b}(t,T) \right\|^2,
 	\end{align*} 
 holds $dt\otimes dQ$-almost surely for $0 \le t \le T \le T^*$ on $\{\tau >t\}$.
\end{enumeratei}   
\end{proposition}

\begin{proof}
We set $E= 1-H$ and $F(t,T)= \exp\left(-\int_{t}^{T} f(t,u)du \right)$. Then \eqref{defaultable price}
can be written as $P(t, T)= E(t) F(t,T).$
Integrating by part yields 
$$dP(t, T) = F(t-,T)dE(t) + E(t-) dF(t,T)+ d[E,F(\cdot,T)]_t. $$

For $\{t < \tau \}$,  
\begin{align*}
dP(t, T){}&= P(t-,T)\left(-\lambda_t dt + \left( f(t,t)+ \frac{1}{2}\left\| \overline{b}(t,T) \right\|^2 -  \overline{a}(t,T) \right)dt\right)\\ 
& + P(t-,T)\left(d M^\lambda + \overline{b}(t,T)dW_t \right).
\end{align*}
The discounted bond price process is a local martingale if and only if the predictable part in the semimartingale decomposition vanishes, i.e., 
\begin{equation}
 f(t,t) - r_t - \lambda_t- \bar{a}(t,T) + \frac{1}{2} \left\| \overline{b}(t,T) \right\|^2 = 0.
 \end{equation}
 Letting $T=t$ we obtain (i) and (ii) and the result follows.
\end{proof}

Next, we derive the no-arbitrage conditions for the forward rate in term of the intensity and the short rate, and also the conditions for the drift and volatility parameters, under homogeneous ambiguity. Set $\lambda^{*}_t:= f(t,t)-r_t$, for $t\in [0,T^*]$. Consider a real-valued, measurable, $\mathbb{F}$-progressive process $\theta=(\theta_t)_{t \geq 0}$ such that the process $z^{\theta} =(z^{\theta}_t)_{0 \le t \le T^*}$ is given as the unique strong solution of  
$$dz^{\theta_t} = - \theta_t z^{\theta_t}dW_t, \quad z^{\theta}_0=1.$$
We assume that $\theta$ is sufficiently integrable, such that $z^\theta$ is a $P'$-martingale
\begin{theorem}
\label{thm:no-arbitrage}
Under Assumption \ref{parameter}, the discounted bond prices are local martingales, if and only if the following conditions are satisfied on $\{\tau > t\}$:  
\begin{enumeratei}
\item there exists an $\bbF$-progressive $\theta^*$ such that $E'[z^{\theta^*}_{T^*}]=1$,
\item the drift condition 
$$
\bar a(t,T) = \half \parallel \bar b(t,T) \parallel^2 - \bar b(t,T) \theta^*_t,
$$
holds $dt\otimes dP'$-almost surely on $\{t<\tau\}$.
\end{enumeratei}
Then there exists an ELMM with respect to $\bar \ccP$.
\end{theorem}

\begin{proof}
Fix $P^\lambda \in \bar \ccP$.
Condition (i) guarantees that $z^{\theta^*}$ is a density process for a change of measure via the Girsanov theorem. We define  
\begin{equation*}
Z^{*}_{T^*} := \begin{cases}
\exp\Big( \int_0^t (1-\lambda^*_s) \lambda_s ds \Big),  & t < \tau \\
\lambda^{*}_\tau \exp\Big( \int_0^\tau (1-\lambda^{*}_s) \lambda_s ds \Big) & t \ge \tau,
\end{cases}
\end{equation*}
being the density process with regards the change in intensity from $\lambda$ to  $\lambda^*$, for $\int_{0}^{t} \lambda^{*}_sds < \infty$, see \cite[Theorem VI.2.T2]{Bremaud1981}. We may define
$$ dP^* := Z^{\lambda^*}_{T^*} z^{\theta^*}_{T^*} dP^\lambda . $$
That is, $P^* \sim P^\lambda $. We now show that $P^*$ is also a local martingale measure. First, note that 
\begin{equation*}
\label{intensityBm}
W_t^{*}:= W_t - \int_{0}^{t}\theta^*_s ds \quad \forall t \in [0,T^*]
\end{equation*}
is a $P^*$-Brownian motion.
Recall for $\{t < \tau \}$,
\begin{align*}
\frac{dP(t, T)}{P(t-,t)}{}&= \left(-\lambda^*(t) + f(t,t)+ \frac{1}{2} \left\| \overline{b}(t,T) \right\|^2-  \overline{a}(t,T) \right)dt\\ 
& + dM^\lambda  - \overline{b}(t,T)dW_t.
\end{align*}
Hence under the change of measure,
\begin{align*}
\frac{dP(t, T)}{P(t-,t)}{}&= \left(-\lambda^*(t) + f(t,t)+ \frac{1}{2} \left\| \overline{b}(t,T) \right\|^2-  \overline{a}(t,T) 
- \bar b(t,T) \theta ^* _t\right)dt\\ 
& + dM^\lambda - \overline{b}(t,T)dW^*_t.
\end{align*}
After discounting with $\gamma$, $\gamma^{-1} P(.,T)$ is a local martingale if and only if the predictable part in its semimartingale decomposition vanishes. Setting $T=t$, condition (ii) together with the definition of $\lambda^*$ holds. Thus, $P^*$ is an ELMM for $P^\lambda$. As $P^\lambda$ is arbitrary, $P^*$ is an ELMM with respect to $\bar \ccP$ and we conclude. 
\end{proof}

\subsection{Recovery of market value}
\label{sec:recovery}
In reduced-form models, there are some recovery assumptions, such as the zero recovery, fractional recovery of treasury, fractional recovery of par value, see \cite[Chapter 8]{BieleckiRutkowski2002} for detail. We have so far considered the case where the credit risky bond becomes worthless and there is zero recovery as soon as default event occurs. Here, we will consider the {\em fractional recovery of market value} where the credit risky bond looses a fraction of its market value. We assume that there is ambiguity on the recovery process. The goal is to obtain the necessary and sufficient conditions for the existence of an ELMM for the family $\{(P_{R}(t,T))_{0\leq t \leq T}; T\in [0, T^*]\}$ with respect to the numeraire $\gamma = \exp \left(\int_0^{\cdot} r_t dt \right)$ and the set of probability measures $\bar{\ccP}$. Thus, extending Theorem \ref{thm:no-arbitrage} to general recovery schemes.
To this end, we assume that on the given probability space $(\Omega,\ccG,P')$, there is additionally a marked point process $(T_n,R_n)_{n \ge 1}$ where the random times $T_n\to \infty$ as $n \to \infty$, which is independent of $W$ and $\tau$ under $P'$. The associated \emph{recovery process} is denoted by
$$ R_t = \prod_{T_n \le t} R_n .$$ 
We assume that $0<T_1< \cdots$ are the jumping times from a Poisson process with intensity one, the recovery values $(R_n)$ are i.i.d.\ with uniform distribution in $[\underline r, \bar r]\subset(0,1]$ (independent from the jumping times). Then,  $R=(R_t)_{t \ge 0}$ is non-increasing and $R_t>0$ for all $t \ge 0$. 
 
The filtration $\mathbb G=(\ccG_t)_{0 \le t \le T^*}$ is in analogy to the setting of Section \ref{sec:Ambiguity} and is obtained by a progressive enlargement with default information (in this case $R$), i.e.,
\begin{align*}
	\ccF_t = \bigcap_{\epsilon >0} \sigma( R_s , W_s: 0 \le s \le t+\epsilon), \qquad 0 \le t \le T^*.
\end{align*}
We assume that $\ccG=\ccG_{T^*}$. As $R$ is a $\bbG$-submartingale which is stochastically continuous, there is a multiplicative Doob-Meyer decomposition, i.e.,\ there exists a $\bbG$-predictable, positive process $h$, such that 
$$ R_t e^{\int_0^t h_s ds}, \quad t \ge 0$$
is a $\bbG$-martingale. Here the process $e^{\int_0^\cdot h_s ds}$ is the \emph{exponential compensator} of $R$ (indeed it is a simple exercise to compute $h$).

Again, we define admissible densities with respect to $\bar \ccP$ through
$$ \bar \ccA_R := \{ h^*: h^* \text{ is }\bbF\text{-predictable and }E^{P}[Z^{h^*}_{T^*}]<\infty\text{ for all }P \in \bar \ccP\}. $$
The associated densities $Z^{h^*}$ for $h^*\in\bar \ccA_R$ are the possible densities for equivalent measure changes when starting from any measure $P\in\bar \ccP$.  

Under this assumption of fractional recovery of market value, the term structure of credit risky bond prices can be assumed to be of the form 
\begin{equation}
\label{Recoverybondprice}
P_{R}(t, T) = R_t \exp\left(-\int_{t}^{T} f(t,u)du \right), \qquad  0 \leq t \leq T \le T^*.
\end{equation}

\begin{remark}
If a default occurs at $t$, the bond loses a random fraction $q_{t}=1-R_t$ of its pre-default value, where $(q_s)_{[0,T^*]}$ is a predictable process with values in $[a,b] \in [0,1)$. Thus, the value of $(1-q_{t})P(t-, T)$ is immediately available to the bond owner at default. It is still subject to default risk because of the possibly following defaults given by $\{ T_n:T_n> t \}$.
\end{remark}
 
\begin{theorem}
\label{ambiguity_recovery}
Let $h^*_t := f(t,t)-r_t, \text{ for } t \in [0, T^*]$. Assume that Assumption \ref{parameter} holds and 
\begin{enumeratei}
\item there exists an $\bbF$-progressive $\theta^*$ such that $E'[z^{\theta^*}_{T^*}]=1$,
\item the drift condition 
$$\bar a(t,T) = \half \parallel \bar b(t,T) \parallel^2 - \bar b(t,T) \theta^*_t,
$$
holds $dt\otimes dP'$-almost surely on $\{t<\tau\}$.
\end{enumeratei}
Then there exists an ELMM with respect to $\bar \ccP$.
\end{theorem}

\begin{proof}
Fix $P^{\lambda} \in \bar \ccP$.
$Q^* \sim P^{\lambda}$ if $$ dQ^* := Z^{h^*}_{T^*} z^{\theta^*}_{T^*} dP^{\lambda},$$ 
and 
\begin{equation*}
Z^{*}_{T^*} := \begin{cases}
\exp\Big( \int_0^t (1-h^*_s) h_s ds \Big),  & t < \tau \\
\lambda^{*}_\tau \exp\Big( \int_0^\tau (1-h^{*}_s) h_s ds \Big) & t \ge \tau.
\end{cases}
\end{equation*}
is the density process with regards the change in intensity from $h$ to  $h^*$, for $\int_{0}^{t} h^{*}_sds < \infty$.
We now show that $Q^*$ is also a local martingale measure. 
Recall, 
By definition of $(h^*_s)_{s\geq 0}$, we have that $R_t e^{\int_0^t h_s^* ds},$ $t \ge 0$ is a $\bbG$-martingale, which implies that 
 $$dM_t= e^{\int_0^\cdot h^*_s ds}(R_{t-}h^*_t dt+dR_t)$$
is the differential of a $\bbG$-martingale. Set $F(t,T)= \exp\left(-\int_{t}^{T} f(t,u)du \right)$. Then \eqref{Recoverybondprice}
can be written as $P_R(t, T)= R(t) F(t,T).$ Integrating by part yields 
$$dP_R(t, T) = F(t,T)dR_t + R(t-) dF(t,T)+ d[R,F(\cdot,T)]_t=(1)+(2)+(3).$$
For $\{t < \tau \}$,  
$$(1)= F(t,T)dR_t=P_R(t-, T) ( (R_{t-})^{-1} e^{-\int_0^t h_s^* ds}dM_t -h^*_tdt). $$
$$(2)=  R(t-) dF(t,T)= P_R(t-, T)\left(\left( f(t,t)+ \frac{1}{2} \left\| \overline{b}(t,T) \right\|^2-  \overline{a}(t,T) \right)dt- \overline{b}(t,T)dW_t\right).$$
$(3)=0$. Thus, 
\begin{align*}
\frac{dP_{R}(t, T)}{P_{R}(t-,t)}{}&= \left(-h^*(t) + f(t,t)+ \frac{1}{2} \left\| \overline{b}(t,T) \right\|^2-  \overline{a}(t,T) \right)dt\\ 
& + \left(\frac{e^{-\int_0^t h^* ds}}{R_{t-}}\right)dM_t - \overline{b}(t,T)dW_t.
\end{align*}
Introducing the change of measure on the Brownian motion, 
\begin{align*}
\frac{dP_{R}(t, T)}{P_{R}(t-,t)}{}&= \left(-h^*(t) + f(t,t)+ \frac{1}{2} \left\| \overline{b}(t,T) \right\|^2-  \overline{a}(t,T)-\overline{b}(t,T)\theta^* \right)dt\\ 
& + \left(\frac{e^{-\int_0^t h^* ds}}{R_{t-}}\right)dM_t - \overline{b}(t,T)dW^{*}_t.
\end{align*}
After discounting with $\gamma$, $\gamma^{-1} P_R(.,T)$ is a local martingale if and only if the predictable part is zero, that is, 
\begin{equation*}
\label{driftcondition}
-r_t-h^*(t) + f(t,t)+ \frac{1}{2} \left\| \overline{b}(t,T) \right\|^2-  \overline{a}(t,T)-\overline{b}(t,T)\theta^*  =0 \quad \forall t\leq T.
\end{equation*}
This is needed only for $t\leq \tau$. This is due to the assumption that the recovery value is instantaneously paid to the bond holder.
Since the above equation hold for $t\leq \tau \wedge T$ and 
\begin{equation*}
\frac{1}{2} \left\| \overline{b}(t,T) \right\|^2-  \overline{a}(t,T)-\overline{b}(t,T)\theta^* =0
\end{equation*}
 if $T=t$, condition (ii) together with the definition of $h^*$ holds, and the results follows. 
\end{proof}

\section{Robust bond pricing interval}
\label{sec:Robust}  
Under the assumption of (robust) no-arbitrage as formalized in Definition \ref{def:arbitrage}, the  (zero-recovery) bond price at time $t$ which pays a unit at maturity $T$ is given by an expectation under an ELMM $Q$ due to Theorem \ref{thm:FTAP-NFLV_homogeneous}. Hence, 
\begin{equation}
\label{eq:Bond-price}
P(t, T)= \ind{\tau > t}E^Q\left[e^{-\int_{t}^{T}(r_s + \lambda^{*} _s)ds}| \ccF_t \right], \qquad 0 \le t \le T;
\end{equation}
here $\lambda^*$ lies necessarily in $[\underline \lambda, \overline \lambda]$. Our goal here is to specify a polynomial process  which satisfies this requirement and to provide pricing formulas for the computation of the expectation in \eqref{eq:Bond-price}.

Following the work of \cite{Delbaen2002}, let us assume that $\lambda^*$ is the unique strong solution of the SDE
\begin{equation}
\label{eq:Jacobi}
d\lambda^{*}_t = \alpha (\lambda_{\mu} - \lambda^{*}_t)dt + \beta \sqrt{(\lambda^{*}_t -\underline{\lambda})(\bar{\lambda} -\lambda^{*}_t) }dW_t, \quad  \lambda_0 \in [\underline{\lambda}, \bar{\lambda}].
\end{equation}
Here, $W$ is a Brownian motion on the probability space $(\Omega, \ccF, Q )$ endowed with the canonical filtration $\mathbb{F}=(\ccF_t)_{0\leq t\leq T^*}$ generated by $W$ that satisfies the usual conditions. We assume that $\alpha, \beta >0$ and $\underline{\lambda} < \lambda_{\mu} < \bar{\lambda}$ which guarantee the existence of a stationary distribution. By definition, the drift function $\mu (x)= \alpha (\lambda_{\mu} - x)$ is Lipschitz continuous.
 Let the volatility function $\sigma (x)= \beta \sqrt{(x -\underline{\lambda})(\bar{\lambda} -x)}$, then for $x,y \in [\underline{\lambda}, \bar{\lambda}]$, 
$$|\sigma (x)^2 - \sigma (y)^2| = \beta |\underline{\lambda} + \bar{\lambda}-(x+y)||x-y|\leq \beta  (\bar{\lambda} - \underline{\lambda})|x-y|.$$
Thus, $\sigma (x)$ is H\"older-$\frac{1}{2}$ continuous. The continuity properties of $\mu (x)$ and $\sigma (x)$  guaranteed the pathwise uniquessness of \eqref{eq:Jacobi} using the general uniqueness theorem (Theorem 4.5, \cite{RevuzYor}). The following pricing formula was obtained in Theorem 3.1 in \cite{Delbaen2002}. Let 
$$ B(t,T) = E^Q\left[e^{-\int_{t}^{T} \lambda^{*}_s ds}|\lambda^{*}_t= \lambda \right]. $$

\begin{theorem}
\label{Bond_Delbaen} Under \eqref{eq:Jacobi} it holds that
\begin{align*}
B(t,T)
= & e^{-\underline{\lambda} (T-t)} \begin{cases} 1+ \sum_{n=1}^{\infty} (\bar{\lambda}-\underline{\lambda})^{n} \cdots\\ \sum_{(v_n,\cdots, v_1) \in V^n} \psi_{v_n}(\frac{\lambda-\underline{\lambda}}{\bar{\lambda}-\underline{\lambda}}) \prod_{j=n}^{1}k_{v_{j}}q (v_j-v_{j-1})I_{t,T}^{n}(y_{v_n},\cdots, y_{v_1})\end{cases}
\end{align*}
where 
\begin{align*}
V^n & = \{(v_n, \cdots, v_1) \in {\mathbb Z}_{+}^{n}: |v_j-v_{j-1}| \leq 1, 1 \leq j \leq n, v_0=0\},  \\
q(v_j, v_{j-1}) &= \begin{cases} \frac{(2v(a+b+v-1)+a(a+b-2))\Gamma^{2}(a)v! \Gamma(b+v)}{(a+b+2v-1)(a+b+2v-2)(a+b+2v-3)\Gamma(a+v-1)\Gamma(a+b+v-2)}&\quad \text{  if } v_j= v_{j-1}\\ -\frac{v! \Gamma^2(a)\Gamma(b+v)}{(a+b+2v-1)(a+b+2v-2)(a+b+2v-3)\Gamma(a+v-1)\Gamma(a+b+v-2)} & \quad\text{  if }|v_j -v_{j-1}|=1\end{cases}\\
\text{ for } v= v_j \vee v_{j-1}, \\
I_{t,T}^{n}(y_{v_n},\cdots, y_{v_1})&= \int_{t}^{T}\int_{s_n}^{T} \cdots \int_{s_2}^{T} \exp\{-\sum_{j=n}^{1}y_j(s_j-s_{j+1}) \} ds_1 \cdots ds_n,
\end{align*}
with $s_{n+1}=t$.
\end{theorem}

The bond price can now be approximated by the truncated sum of series from  Theorem \ref{Bond_Delbaen}, i.e.~
\begin{align*}
B^j(t, T):=
 & e^{-\underline{\lambda} (T-t)} \begin{cases} 1+ \sum_{n=1}^{j} (\bar{\lambda}-\underline{\lambda})^{n} \cdots\\ 
 \sum_{(v_n,\cdots, v_1) \in V^n} \psi_{v_n}(\frac{\lambda-\underline{\lambda}}{\bar{\lambda}-\underline{\lambda}}) \prod_{j=n}^{1}k_{v_{j}}q (v_j-v_{j-1})I_{t,T}^{n}(y_{v_n},\cdots, y_{v_1})\end{cases}.
\end{align*}
 
Proposition 4.1 \cite{Delbaen2002} shows that the truncated sum up to second order, and it turns out that the volatility coefficient $\beta$ appears only for $j=2$. Thus, one should consider at least $j=2$, i.e.~$P^{2}(t,T)$ to take into account the volatility coefficient in the approximation result.  $P^{0}= e^{-\underline{\lambda} (T-t)}$ is an obvious upper bound of the no-arbitrage bond price for any initial default intensity. Since the default intensity has a bounded support, one can as well derive the lower and the upper bounds for the no-arbitrage bond prices. The following result is Theorem 4.2 in \cite{Delbaen2002}.

\begin{theorem}
\label{thm:bond-bounds}
\begin{enumeratei}
\item Lower bound: 
\begin{equation*}
\exp\left( -\lambda_{\mu} (T-t)-(\lambda -\lambda_{\mu})\frac{1-e^{-\alpha(T-t)}}{\alpha} \right) \leq B(t,T)
\end{equation*}
\item Upper bound 
\begin{align*}
B(t,T) \leq \left(1-\gamma -(z-\gamma)) \frac{1-e^{-\alpha (T-t)}}{\alpha} \right) e^{-\underline{\lambda}(T-t)}\\
+ \left(\gamma + (z-\gamma) \frac{1-e^{-\alpha (T-t)}}{\alpha} \right) e^{-\bar{\lambda}(T-t)},
\end{align*}
where $z= \frac{\lambda -\underline{\lambda}}{\bar{\lambda}-\underline{\lambda}}$ and $\gamma = \frac{\lambda_{\mu}-\underline{\lambda}}{\bar{\lambda}-\underline{\lambda}}$.
\end{enumeratei}
\end{theorem}

\begin{remark}[Recovery] 
According to Theorem \ref{ambiguity_recovery}, one obtains an immediate generalization to fractional recovery of market value when replacing $\lambda^*$ in the above calculations by $h^*$.
\end{remark}

\end{document}